\documentclass{article}
\usepackage[utf8]{inputenc}
\usepackage[T1]{fontenc}
\usepackage{cite}
\usepackage{fullpage}
\usepackage{amssymb}
\usepackage{amsmath}
\usepackage{graphicx}
\usepackage{xspace}
\usepackage{hyperref}
\hypersetup{
  colorlinks,%
  citecolor=black,%
  filecolor=black,%
  linkcolor=black,%
  urlcolor=black
  }

\usepackage{enumerate}

\usepackage[noend]{algorithmic}
\usepackage{algorithm}

\usepackage[inline,nomargin,draft]{fixme}
\fxsetface{inline}{\bfseries\color{red}}

\usepackage{tikz}
\usetikzlibrary{calc}
\usetikzlibrary{shapes.misc}

\newcommand{\LCE}{\ensuremath{\mathrm{LCE}}\xspace}

\newcommand{\occ}{\ensuremath{\mathrm{occ}}}

\newtheorem{theorem}{Theorem}
\newtheorem{lemma}{Lemma}

\newtheorem{example}{Example}
\newtheorem{definition}{Definition}

\newcommand{\qed}{\hfill\ensuremath{\Box}\medskip\\\noindent}
\newenvironment{proof}{\noindent\emph{Proof. }}{\qed}

\newcommand{\infrac}[2]{{#1}/{#2}}

\newcommand{\floor}[1]{\left\lfloor{#1}\right\rfloor}

\newcommand{\res}{\textit{res}}

\title{Time-Space Trade-Offs for\\ Longest Common Extensions\thanks{This work was partly supported by EPSRC. An extended abstract of this paper appeared in the proceedings of the 23rd Annual Symposium on Combinatorial Pattern Matching.}}
\author{Philip Bille\footnotemark[2] \and Inge Li G\o rtz\footnotemark[2] \and Benjamin Sach\footnotemark[3] \and Hjalte Wedel Vildh{\o}j\footnotemark[2]\\[1em]
\footnotemark[2]~~Technical University of Denmark, DTU Compute, \texttt{\{phbi,inge\}@dtu.dk,hwv@hwv.dk}\\
\footnotemark[3]~~University of Warwick, Department of Computer Science, \texttt{sach@dcs.warwick.ac.uk}
}
\date{}
\begin{document}
\maketitle

\begin{abstract}
We revisit the longest common extension (LCE) problem, that is, preprocess a string $T$ into a compact data structure that supports fast LCE queries. An LCE query takes a pair $(i,j)$ of indices in $T$ and returns the length of the longest common prefix of the suffixes of $T$ starting at positions $i$ and $j$. We study the time-space trade-offs for the problem, that is, the space used for the data structure vs. the worst-case time for answering an LCE query. Let $n$ be the length of $T$. Given a parameter $\tau$, $1 \leq \tau \leq n$, we show how to achieve either $O(\infrac{n}{\sqrt{\tau}})$ space and $O(\tau)$ query time, or $O(\infrac{n}{\tau})$ space and $O(\tau \log({|\LCE(i,j)|}/{\tau}))$ query time, where $|\LCE(i,j)|$ denotes the length of the LCE returned by the query. These bounds provide the first smooth trade-offs for the LCE problem and almost match the previously known bounds at the extremes when $\tau=1$ or $\tau=n$. We apply the result to obtain improved bounds for several applications where the LCE problem is the computational bottleneck, including approximate string matching and computing palindromes. We also present an efficient technique to reduce LCE queries on two strings to one string. Finally, we give a lower bound on the time-space product for LCE data structures in the non-uniform cell probe model showing that our second trade-off is nearly optimal.
\end{abstract}

\section{Introduction}
Given a string $T$, the \emph{longest common extension} of suffix $i$ and $j$, denoted $\LCE(i,j)$, is the length of the longest common prefix of the suffixes of $T$ starting at position $i$ and $j$. The \emph{longest common extension problem} (LCE problem) is to preprocess $T$ into a compact data structure supporting fast longest common extension queries.

The LCE problem is a basic primitive that appears as a subproblem in a wide range of string matching problems such as approximate string matching and its variations~\cite{LV1989, Myers1986,LMS1998, CH2002, ALP2004}, computing exact or approximate tandem repeats \cite{Main1984, Landau2001, GS2004}, and computing palindromes. In many of the applications, the LCE problem is the computational bottleneck.

In this paper we study the time-space trade-offs for the LCE problem, that is, the space used by the preprocessed data structure vs. the worst-case time used by LCE queries. We assume that the input string is given in read-only memory and is not counted in the space complexity. There are essentially only two time-space trade-offs known: At one extreme we can store a suffix tree combined with an efficient nearest common ancestor (NCA) data structure~\cite{HT1984} (other combinations of $O(n)$ space data structures for the string can also be used to achieve this bound, e.g.~\cite{FH2006}). This solution uses $O(n)$ space and supports LCE queries in $O(1)$ time. At the other extreme we do not store any data structure and instead answer queries simply by comparing characters from left-to-right in $T$. This solution uses $O(1)$ space and answers an $\LCE(i,j)$ query in $O(|\LCE(i,j)|) = O(n)$ time. 
This approach was recently shown to be very practical~\cite{INT2010}.

\subsection{Our Results}
We show the following main result for the longest common extension problem.

\begin{theorem}\label{thm:main}
For a string $T$ of length $n$ and any parameter $\tau$, $1 \leq \tau \leq n$, $T$ can be preprocessed into a data structure supporting $\LCE(i,j)$ queries on $T$. This can be done such that the data structure
\begin{itemize}
\item[(i)] uses $O\big(\frac{n}{\sqrt{\tau}}\big)$ space and supports queries in $O(\tau)$ time. The preprocessing of $T$ can be done in $O(\frac{n^2}{\sqrt{\tau}})$ time and $O\big(\frac{n}{\sqrt{\tau}}\big)$ space.
\item[(ii)] uses $O\left(\frac{n}{\tau}\right)$ space and supports queries in $O\left(\tau \log \left(\frac{|\LCE(i,j)|}{\tau} \right)\right)$ time. The preprocessing of $T$ can be done in $O(n)$ time and $O(\frac{n}{\tau})$ space. The solution is randomised (Monte-Carlo); with high probability, all queries are answered correctly.
\item[(iii)] uses $O\big(\frac{n}{\tau}\big)$ space and supports queries in $O\left(\tau \log \left(\frac{|\LCE(i,j)|}{\tau} \right)\right)$ time. The preprocessing of $T$ can be done in $O(n\log n)$ time and $O(n)$ space. The solution is randomised (Las-Vegas); the preprocessing time bound is achieved with high probability.
\end{itemize}
\end{theorem}
Unless otherwise stated, the bounds in the theorem are worst-case, and with high probability means with probability at least $1-1/n^c$ for any constant $c$.

Our results provide a smooth time-space trade-off that allows several new and non-trivial bounds. For instance, with $\tau = \sqrt{n}$ \autoref{thm:main}(i), gives a solution using $O(n^{3/4})$ space and $O(\sqrt{n})$ time. If we allow randomisation, we can use \autoref{thm:main}(iii) to further reduce the space to $O(\sqrt{n})$ while using query time $O(\sqrt{n}\log(|\LCE(i,j)|/\sqrt{n})) = O(\sqrt{n}\log n)$. Note that at both extremes of the trade-off ($\tau = 1$ or $\tau=n$) we almost match the previously known bounds. In the conference version of this paper~\cite{Bille2012}, we mistakenly claimed the preprocessing space of \autoref{thm:main}(iii) to be $O(n/\tau)$ but it is in fact $O(n)$. It is possible to obtain $O(n/\tau)$ preprocessing space by using $O(n\log n + n\tau)$ preprocessing time. For most applications, including those mentioned in this paper, this issue have no implications, since the time to perform the LCE queries typically dominates the preprocessing time.

Furthermore, we also consider LCE queries between two strings, i.e.\@ the pair of indices to an LCE query is from different strings. We present a general result that reduces the query on two strings to a single one of them. When one of the strings is significantly smaller than the other, we can combine this reduction with \autoref{thm:main} to obtain even better time-space trade-offs.

Finally, we give a reduction from \emph{range minimum queries} that shows that any data structure using $O(n/\tau)$ bits space in addition to the string $T$ must use at least $\Omega(\tau)$ time to answer an LCE query. Hence, the time-space trade-offs of \autoref{thm:main}(ii) and \autoref{thm:main}(iii) are almost optimal.

\subsection{Techniques}
The high-level idea in \autoref{thm:main} is to combine and balance out the two extreme solutions for the LCE problem. For \autoref{thm:main}(i) we use \emph{difference covers} to sample a set of suffixes of $T$ of size $O(n/\sqrt{\tau})$. We store a compact trie combined with an NCA data structure for this sample using $O(n/\sqrt{\tau})$ space. To answer an LCE query we compare characters from $T$ until we get a mismatch or reach a pair of sampled suffixes, which we then immediately compute the answer for. By the properties of difference covers we compare at most $O(\tau)$ characters before reaching a pair of sampled suffixes. Similar ideas have previously been used to achieve trade-offs for suffix array and LCP array construction~\cite{Kar2006,Pug2008}.

For \autoref{thm:main}(ii) and \autoref{thm:main}(iii) we show how to use Rabin-Karp fingerprinting~\cite{KR1987} instead of difference covers to reduce the space further. We show how to store a sample of $O(\infrac{n}{\tau})$ fingerprints, and how to use it to answer LCE queries using doubling search combined with directly comparing characters. This leads to the output-sensitive $O(\tau \log(\infrac{|\LCE(i,j)|}{\tau}))$ query time. We reduce space compared to \autoref{thm:main}(i) by computing fingerprints on-the-fly as we need them. Initially, we give a Monte-Carlo style randomised data structure (\autoref{thm:main}(ii)) that may answer queries incorrectly. However, this solution uses only $O(n)$ preprocessing time and is therefore of independent interest in applications that can tolerate errors. To get the error-free Las-Vegas style bound of \autoref{thm:main}(iii) we need to verify the fingerprints we compute are collision free; i.e.\@ two fingerprints are equal if and only if the corresponding substrings of $T$ are equal. The main challenge is to do this in only $O(n \log n)$ time. We achieve this by showing how to efficiently verify fingerprints of composed samples which we have already verified, and by developing a search strategy that reduces the fingerprints we need to consider.

Finally, the reduction for LCE on two strings to a single string is based on a simple and compact encoding of the larger string using the smaller string. The encoding could be of independent interest in related problems, where we want to take advantage of different length input strings.

\subsection{Applications}
With \autoref{thm:main} we immediately obtain new results for problems based on LCE queries. We review some the most important below.

\subsubsection{Approximate String Matching}
Given strings $P$ and $T$ and an error threshold $k$, the \emph{approximate string matching problem} is to report all ending positions of substrings of $T$ whose \emph{edit distance} to $P$ is at most $k$. The edit distance between two strings is the minimum number of insertions, deletions, and substitutions needed to convert one string to the other. Let $m$ and $n$ be the lengths of $P$ and $T$. The Landau-Vishkin algorithm~\cite{LV1989} solves approximate string matching using $O(nk)$ LCE queries on $P$ and substrings of $T$ of length $O(m)$. Using the standard linear space and constant time LCE data structure, this leads to a solution using $O(nk)$ time and $O(m)$ space (the $O(m)$ space bound follows by the standard trick of splitting $T$ into overlapping pieces of size $O(m)$). If we plug in the results from \autoref{thm:main} we immediately obtain the following result.
\begin{theorem}\label{thm:approximatematching}
Given strings $P$ and $T$ of lengths $m$ and $n$, respectively, and a parameter $\tau$, $1 \leq \tau \leq m$, we can solve approximate string matching
\begin{itemize}
\item[(i)] in $O\big(\frac{m}{\sqrt{\tau}}\big)$ space and $O(nk\cdot\tau + \frac{nm}{\sqrt{\tau}})$ time, or
\item[(ii)] in $O\big(\frac{m}{\tau}\big)$ space and $O(nk\cdot\tau\log m)$  time with high probability.
\end{itemize}
\end{theorem}
For instance for $\tau=(\frac{m}{k})^{2/3}$ \autoref{thm:approximatematching}(i) gives a solution using $O(n m^{2/3} k^{1/3})$ time and $O(m^{2/3}k^{1/3})$ space. To the best of our knowledge these are the first non-trivial bounds for approximate string matching using $o(m)$ space.

\subsubsection{Palindromes}
Given a string $T$ the \emph{palindrome problem} is to report the set of all \emph{maximal palindromes} in $T$. A substring $T[i \ldots j]$ is a maximal palindrome iff $T[i\ldots j] = T[i \ldots j]^R$ and $T[i -1 \ldots j+1] \neq T[i-1 \ldots j+1]^R$. Here $T[i \ldots j]^R$ denotes the reverse of $T[i \ldots j]$. Any palindrome in $T$ occurs in the middle of a maximal palindrome, and thus the set of maximal palindromes compactly represents all palindromes in $T$. The palindrome problem appears in a diverse range of applications, see e.g.\@~\cite{Allouche2003, Jeuring1994, Matsubara2009, Breslauer1995, Lu2007, Kolpakov2008, Gusfield1997}.

We can trivially solve the problem in $O(n^2)$ time and $O(1)$ space by a linear search at each position in $T$ to find the maximal palindrome. With LCE queries we can immediately speed up this search. Using the standard $O(n)$ space and constant time solution to the LCE problem this immediately implies an algorithm for the palindrome problem that uses $O(n)$ time and space (this bound can also be achieved without LCE queries~\cite{Manacher1975}). Using \autoref{thm:main} we immediately obtain the following result.
\begin{theorem}\label{thm:palindromes}
Given a string of length $n$ and a parameter $\tau$, $1 \leq \tau \leq n$, we can solve the palindrome problem
\begin{itemize}
\item[(i)] in $O\big(\frac{n}{\sqrt{\tau}}\big)$ space and $O\big(\frac{n^2}{\sqrt{\tau}} + n\tau\big)$ time.
\item[(ii)] in $O\big(\frac{n}{\tau}\big)$ space and $O(n\cdot \tau \log n)$ time with high probability.
\end{itemize}
\end{theorem}
For $\tau = \omega(1)$, these are the first sublinear space bounds using $o(n^2)$ time. For example, for $\tau=n^{2/3}$ \autoref{thm:palindromes}(i) gives a solution using $O(n^{5/3})$ time and $O(n^{2/3})$ space. Similarly, we can substitute our LCE data structures in the LCE-based variants of palindrome problems such as \emph{complemented palindromes}, \emph{approximate palindromes}, or \emph{gapped palindromes}, see e.g.\@~\cite{Kolpakov2008}.

\subsubsection{Tandem Repeats} Given a string $T$, the \emph{tandem repeats problem} is to report all squares, i.e.\@ consecutive repeated substrings in $T$. Main and Lorentz~\cite{Main1984} gave a simple solution for this problem based on LCE queries that achieves $O(n)$ space and $O(n\log n + \occ)$ time, where $\occ$ is the number of tandem repeats in $T$. Using different techniques Gąsieniecs et al.~\cite{Gasieniec2005} gave a solution using $O(1)$ space and $O(n\log n + \occ)$ time. Using \autoref{thm:main} we immediately obtain the following result. 
\begin{theorem}
Given a string of length $n$ and parameter $\tau$, $1 \leq \tau \leq n$, we can solve the tandem repeats problem
\begin{itemize}
\item[(i)] in $O\big(\frac{n}{\sqrt{\tau}}\big)$ space and $O\big(\frac{n^2}{\sqrt{\tau}} + n\tau\cdot\log n + occ\big)$ time.
\item[(ii)] in $O\big(\frac{n}{\tau}\big)$ space and $O\big( n\tau\cdot \log^2 n + \occ\big)$ time with high probability.
\end{itemize}
\end{theorem}
While this does not improve the result by Gąsieniecs et al.~\cite{Gasieniec2005} it provides a simple LCE-based solution. Furthermore, our result generalises to the approximate versions of the tandem repeats problem, which also have solutions based on LCE queries~\cite{Landau2001}.

\section{The Deterministic Data Structure}
We now show \autoref{thm:main}(i). Our deterministic time-space trade-off is based on sampling suffixes using \emph{difference covers}.

\subsection{Difference Covers}
A \emph{difference cover modulo $\tau$} is a set of integers $D \subseteq \{0,1,\ldots, \tau-1\}$ such that for any distance $d \in \{0,1,\ldots, \tau-1\}$, $D$ contains two elements separated by distance $d$ modulo $\tau$ (see \autoref{ex:differencecover}).
\begin{example}\label{ex:differencecover}
The set $D=\{1,2,4\}$ is a difference cover modulo $5$.
\begin{center}
\begin{tabular}{cp{1cm}c}
\begin{tabular}{|c|c|c|c|c|c|} \hline
$d$ & $0$ &  $1$ & $2$ & $3$ & $4$ \\ \hline
 $i,j$ & $1,1$ &  $2,1$ & $1,4$ & $4,1$ & $1,2$ \\ \hline
\end{tabular}
&
&
$\vcenter{\hbox{\includegraphics{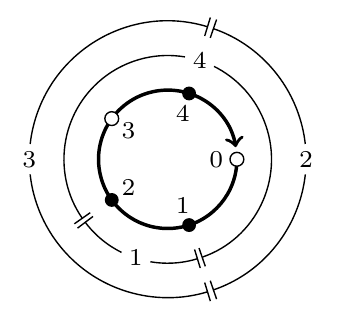}}}$ \\
\end{tabular}
\end{center}
\end{example}
A difference cover $D$ can cover at most $|D|^2$ differences, and hence $D$ must have size at least $\sqrt{\tau}$. We can also efficiently compute a difference cover within a constant factor of this bound.
\begin{lemma}[Colbourn and Ling~\cite{Col2000}]\label{lem:differencecover}
For any $\tau$, a difference cover modulo $\tau$ of size at most $\sqrt{1.5\tau}+6$ can be computed in $O(\sqrt{\tau})$ time.
\end{lemma}

\subsection{The Data Structure}
Let $T$ be a string of length $n$ and let $\tau$, $1 \leq \tau \leq n$, be a parameter. Our data structure consists of the compact trie of a sampled set of suffixes from $T$ combined with a NCA data structure. The sampled set of suffixes $\mathcal{S}$ is the set of suffixes obtained by overlaying a difference cover on $T$ with period $\tau$, that is,
\[\mathcal S = \left \{ i ~\mid~ 1 \leq i \leq n ~\wedge~ i \bmod{\tau} \in D \right \} \; .\]
\begin{example}\label{ex:suffixsample}
Consider the string $T = \textrm{dbcaabcabcaabcac}$. As shown below, using the difference cover from \autoref{ex:differencecover}, we obtain the suffix sample $\mathcal S = \{ 1,2,4,6,7,9,11,12,14,16 \}$.
\begin{center}
\includegraphics{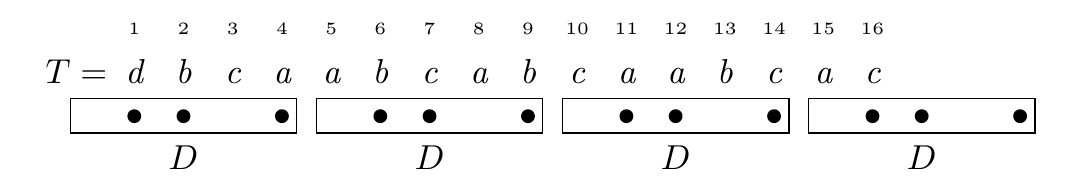}
\end{center}
\end{example}
By \autoref{lem:differencecover} the size of $\mathcal S$ is $O(\infrac{n}{\sqrt{\tau}})$. Hence the compact trie and the NCA data structures use $O(\infrac{n}{\sqrt{\tau}})$ space. We construct the data structure in $O(\infrac{n^2}{\sqrt{\tau}})$ time by inserting each of the $O(\infrac{n}{\sqrt{\tau}})$ sampled suffixes in $O(n)$ time.

To answer an $\LCE(i,j)$ query we explicitly compare characters starting from $i$ and $j$ until we either get a mismatch or we encounter a pair of sampled suffixes. If we get a mismatch we simply report the length of the LCE. Otherwise, we do a NCA query on the sampled suffixes to compute the LCE. Since the distance to a pair of sampled suffixes is at most $\tau$ the total time to answer a query is $O(\tau)$.
This concludes the proof of \autoref{thm:main}(i).

\section{The Monte-Carlo Data Structure}\label{sec:MonteCarlo}
We now show \autoref{thm:main}(ii) which is an intermediate step towards proving \autoref{thm:main}(iii) but is also of independent interest, providing a Monte-Carlo time-space trade-off. The technique is based on sampling suffixes using \emph{Rabin-Karp fingerprints}. These fingerprints will be used to speed up queries with large $\LCE$ values while queries with small $\LCE$ values will be handled naively.

%

\subsection{Rabin-Karp fingerprints}

 Rabin-Karp fingerprints are defined as follows. Let $2n^{c+4} < p \leq 4n^{c+4}$ be some prime and choose $b \in \mathbb{Z}_p$ uniformly at random.  Let $S$ be any substring of $T$, the fingerprint $\phi(S)$ is given by, \[  \phi(S)=\sum_{k=1}^{|S|} S[k] b^k \bmod p\,. \] \autoref{lem:fingerprint} gives a crucial property of these fingerprints (see e.g.~\cite{KR1987} for a proof). That is with high probability we can determine whether any two substrings of $T$ match in constant time by comparing their fingerprints.

\begin{lemma}\label{lem:fingerprint}
Let $\phi$  be a fingerprinting function picked uniformly at random (as described above). With high probability, 
\begin{alignat}{3} \phi(&T[i\ldots i+\alpha -1])&\ =\ &\phi(&T[j \ldots j+\alpha-1]&) \notag \\ \text{ iff } &T[i \ldots i+ \alpha -1]&\ =\ &&T[j \ldots j+\alpha-1] & \text{\ \ \ for all $i,j,\alpha$.} \end{alignat}

\end{lemma}
%


\subsection{The Data Structure} 
The data structure consists of the fingerprint, $\phi_k$, of each suffix of the form $T[k\tau \ldots n]$ for $0 < k < n/\tau$, i.e.\@ $\phi_k=\phi(T[k \tau \ldots n])$. Note that there are $O(n/\tau)$ such suffixes and the starting points of two consecutive suffix  are $\tau$ characters apart. Since each fingerprint uses constant space the space usage of the data structure is $O(n/\tau)$. The $n/\tau$ fingerprints can be computed in left-to-right order by a single scan of $T$ in $O(n)$ time.

\paragraph{Queries} The key property we use to answer a query is given by \autoref{lem:substring-fingerprint}. 

\begin{lemma}\label{lem:substring-fingerprint}
The fingerprint $\phi(T[i \ldots i+\alpha-1])$ of any substring $T[i \ldots i+\alpha-1]$ can be constructed in $O(\tau)$ time. If $i,\alpha$ are divisible by $\tau$, the time becomes $O(1)$.
\end{lemma}
\begin{proof}
Let $k_1=\lceil i/\tau \rceil$ and $k_2= \lceil(i+\alpha)/\tau \rceil$ and observe that we have $\phi_{k_1}$ and $\phi_{k_2}$ stored.
By the definition of $\phi$, we can compute $\phi(T[k_1 \tau \ldots k_2 \tau-1])= \phi_{k_1} - \phi_{k_2} \cdot b^{(k_2-k_1)\tau} \bmod p$ in $O(1)$ time. If $i, \alpha = 0 \bmod \tau$ then $k_1\tau=i$ and $k_2\tau=i+\alpha$ and we are done. Otherwise, similarly we can then convert $\phi(T[k_1 \tau \ldots k_2 \tau-1])$ into $\phi(T[k_1 \tau-1 \ldots k_2 \tau-2])$ in $O(1)$ time by inspecting $T[k_1 \tau-1]$ and $T[k_2 \tau-1]$. By repeating this final step we obtain  $T[i \ldots i+\alpha-1]$ in $O(\tau)$ time.
\end{proof}
We now describe how to perform a query by using fingerprints to compare substrings. We define $\phi_{k}^{\ell}=\phi(T[k\tau \ldots (k+2^\ell) \tau -1])$ which we can compute in $O(1)$ time for any $k, \ell$ by \autoref{lem:substring-fingerprint}.

First consider the problem of answering a query of the form $\LCE(i \tau, j\tau)$. 
Find the largest $\ell$ such that $\phi_{i}^{\ell}=\phi_{j}^\ell$. When the correct $\ell$ is found convert the query into a new query $\LCE((i+2^\ell) \tau,(j+2^\ell)\tau)$ and repeat. If no such $\ell$ exists, explicitly compare $T[i\tau\ldots (i+1)\tau-1]$ and  $T[j\tau\ldots (j+1)\tau-1]$ one character at a time until a mismatch is found. Since no false negatives can occur when comparing fingerprints, such a mismatch exists. Let $\ell_0$ be the value of $\ell$ obtained for the initial query, $\LCE(i \tau, j\tau)$,   and $\ell_q$ the value obtained during the $q$-th recursion. For the initial query, we search for $\ell_0$ in increasing order, starting with $\ell_0=0$. After recursing, we search for $\ell_q$ in descending order, starting with $\ell_{q-1}$. By the maximality of $\ell_{q-1}$, we find the correct $\ell_q$. Summing over all recursions we have $O(\ell_0)$ total searching time and $O(\tau)$ time scanning $T$. The desired query time follows from observing that by the maximality of $\ell_0$, we have that $O(\tau+\ell_0)= O(\tau+ \log( |\LCE(i\tau,j\tau)|/\tau))$.

%



Now consider the problem of answering a query of the form $\LCE(i \tau, j\tau + \gamma)$ where $0<\gamma<\tau$. By \autoref{lem:substring-fingerprint} we can obtain the fingerprint of any substring in $O(\tau)$ time.
This allows us to use a similar approach to the first case. We find the largest $\ell$ such that $\phi(T[j\tau+\gamma \ldots (j+2^\ell)\tau+\gamma-1])= \phi_{i}^{\ell}$ and convert the current query into a new query, $\LCE((i+2^\ell) \tau,(j+2^\ell)\tau + \gamma)$. As we have to apply \autoref{lem:substring-fingerprint} before every comparison, we obtain a total complexity of $O(\tau \log (|\LCE(i \tau, j\tau + \gamma)|/\tau))$.

In the general case an $\LCE(i,j)$ query can be reduced to one of the first two cases by scanning $O(\tau)$ characters in $T$. By \autoref{lem:fingerprint}, all fingerprint comparisons are correct with high probability and the result follows. 


\section{The Las-Vegas Data Structure}\label{sec:LasVegas}

We now show \autoref{thm:main}(iii). The important observation is that when we compare the fingerprints of two strings during a query in \autoref{sec:MonteCarlo}, one of them is of the form $T[j\tau\ldots j\tau+\tau\cdot 2^\ell -1]$ for some $\ell,j$. Consequently, to ensure all queries are correctly computed, it suffices that $\phi$ is \emph{$\tau$-good}:



\begin{definition}\label{dfn:goodfinger}
A fingerprinting function, $\phi$ is $\tau$-\emph{good} on $T$ iff
\begin{alignat}{3}\label{prop:good} \phi(&T[j\tau\ldots j\tau+\tau\cdot 2^\ell -1])&\ =\ &\phi(&T[i \ldots i+\tau\cdot 2^\ell-1]&) \notag \\ \text{ iff } &T[j\tau \ldots j\tau+ \tau\cdot 2^\ell -1]&\ =\ &&T[i \ldots i+\tau\cdot 2^\ell-1]& \text{\ \ \ for all $(i,j,\ell)$.} \end{alignat}
\end{definition}
In this section we give an algorithm which decides whether a given $\phi$ is $\tau$-good on string $T$. The algorithm uses $O(n)$ space and takes $O(n\log n)$ time with high probability. By using $O(n\log n + n\tau)$ preprocessing space the algorithm can also be implemented to use only $O(n/\tau)$ space. By \autoref{lem:fingerprint}, a uniformly chosen $\phi$ is $\tau$-good with high probability and therefore (by repetition) we can generate such a $\phi$ in the same time/space bounds.
For brevity we assume that $n$ and $\tau$ are powers-of-two.

\subsection{The Algorithm} We begin by giving a brief overview of Algorithm~\ref{alg:ver}. For each value of $\ell$ in ascending order (the outermost loop), Algorithm~\ref{alg:ver} checks (\ref{prop:good}) for all $i,j$. For some outermost loop iteration $\ell$, the algorithm  inserts the fingerprint of each block-aligned substring into a dynamic perfect dictionary, $D_\ell$ (lines 3-9). A substring is block-aligned if it is of the form, $T[j\tau  \ldots (j+2^\ell)\tau-1]$ for some $j$ (and block-unaligned otherwise). If more than one block-aligned substring has the same fingerprint, we insert only the left-most as a representative. For the first iteration, $\ell=0$ we also build an Aho-Corasick automaton~\cite{AC1975}, denoted $AC$, with a pattern dictionary containing every block-aligned substring.

 The second stage (lines 12-21) uses a sliding window technique, checking each time we slide whether the fingerprint of the current ($2^\ell\tau$)-length substring occurs in the dynamic dictionary, $D_\ell$. If so we check whether the corresponding substrings match (if not a collision has been revealed and we abort). For $\ell>0$, we use the fact that (\ref{prop:good}) holds for all $i,j$ with $\ell-1$ (otherwise, Algorithm~\ref{alg:ver} would have already aborted) to perform the check in constant time (line 18). I.e. if there is a collision it will be revealed by comparing the fingerprints of the left-half ($L'_i \neq L_k$) or right-half ($R'_i \neq R_k$) of the underlying strings.  For $\ell=0$, the check is performed using the $AC$ automaton (lines 20-21). We achieve this by feeding $T$ one character at a time into the $AC$. By inspecting the state of the $AC$ we can decide whether the current $\tau$-length substring of $T$ matches any block-aligned substring. 

\begin{algorithm}[ht]
\caption{Verifying a fingerprinting function, $\phi$ on string $T$ \label{alg:ver}}
 \begin{algorithmic}[1]
\STATE \COMMENT{$AC$ is an Aho-Corasick automaton and each $D_\ell$ is a dynamic dictionary}
 \FOR{$\ell = 0 \ldots \log_2 (n/\tau)$}
\STATE \COMMENT{Insert all distinct block-aligned substring fingerprints into $D_\ell$}
\FOR{$j=1 \ldots n/\tau-2^\ell$}

    \STATE $f_j \gets \phi(T[j\tau \ldots (j+2^\ell)\tau-1])$
    \STATE $L_j \gets \phi(T[j\tau \ldots (j+2^{\ell-1})\tau -1])$, $R_j \gets \phi(T[(j+2^{\ell-1})\tau \ldots (j+2^\ell)\tau-1])$
\IF{$\not \exists (f_k,L_k,R_k,k) \in D_\ell$ such that $f_j=f_k$}
\STATE Insert $(f_j,L_j,R_j,j)$ into $D_\ell$ indexed by $f_j$
\STATE {\bf if} {$\ell=0$} {\bf then} Insert $T[j\tau \ldots (j+1)\tau-1]$ into $AC$ dictionary
\ENDIF
\ENDFOR

\STATE \COMMENT{Check for collisions between any block-aligned and unaligned substrings}
\STATE {\bf if} {$\ell=0$} {\bf then} Feed $T[1 \ldots \tau-1]$ into $AC$
\FOR{$i=1 \ldots n-\tau\cdot 2^\ell+1$}
    \STATE $f'_i \gets \phi(T[i \ldots i+\tau\cdot{2^\ell}-1])$
    \STATE $L'_i \gets \phi(T[i \ldots i+\tau\cdot2^{\ell-1}-1])$, $R'_i \gets \phi(T[(i+2^{\ell-1})\tau \ldots i+\tau\cdot2^{\ell}-1])$
    \STATE {\bf if} {$\ell=0$} {\bf then} Feed $T[i+\tau-1]$ into $AC$ \COMMENT $AC$ now points at $T[i \ldots i+ \tau-1]$
    \IF{$\exists (f_k,L_k,R_k,k) \in D_\ell $ such that $f'_i=f_k$}
        \IF{$\ell>0$}

        \STATE {\bf if } {($L'_i \neq L_k$ or $R'_i \neq R_k$)} {\bf then} {\bf abort}
        \ELSE

        \STATE Compare $T[i \ldots i+\tau-1]$ to $T[k\tau \ldots (k+1)\tau-1]$ by inspecting $AC$
        \STATE {\bf if} {$T[i \ldots i+\tau-1] \neq T[k\tau \ldots (k+1)\tau-1]$} {\bf then} {\bf abort}
        \ENDIF

    \ENDIF
\ENDFOR
\ENDFOR
\end{algorithmic}
\end{algorithm}


 \paragraph{Correctness}We first consider all points at which Algorithm~\ref{alg:ver} may abort. First observe that if line 21 causes an abort then (\ref{prop:good}) is violated for $(i,k,0)$. Second, if line 18 causes an abort either $L'_i \neq L_k$ or $R'_i \neq R_k$. By the definition of $\phi$, in either case, this implies that $T[i \ldots i+\tau\cdot{2^\ell}-1] \neq T[k\tau \ldots k\tau+2^{\ell}\tau -1]$. By line 16, we have that $f'_i=f_k$ and therefore (\ref{prop:good}) is violated for $(i,k,\ell)$. Thus, Algorithm~\ref{alg:ver} does not abort if $\phi$ is $\tau$-good.

It remains to show that Algorithm~\ref{alg:ver} always aborts if $\phi$ is not $\tau$-good. Consider the total ordering on triples $(i,j,\ell)$ obtained by stably sorting (non-decreasing) by $\ell$ then $j$ then $i$. E.g. $(1,3,1) < (3,2,3) < (2,5,3) < (4,5,3)$. Let $(i^*,j^*,\ell^*)$ be the (unique) smallest triple under this order which violates (\ref{prop:good}). We first argue that  $(f_{j^*},L_{j^*},R_{j^*},{j^*})$ will be inserted into $D_{\ell^*}$ (and $AC$ if $\ell^*=0$). For a contradiction assume that when Algorithm~\ref{alg:ver} checks for $f_{j^*}$ in $D_{\ell^*}$ (line 7, with $j=j^*,\ell=\ell^*$) we find that some $f_k=f_{j^*}$ already exists in $D_{\ell^*}$, implying that $k<j^*$. If $T[j^*\tau\ldots j^*\tau+\tau 2^\ell -1] \neq T[k\tau\ldots k\tau+\tau2^\ell -1]$ then $(j^*\tau,k,\ell^*)$ violates (\ref{prop:good}). Otherwise, $(i^*,k,\ell^*)$ violates (\ref{prop:good}). In either case this contradicts the minimality of $(i^*,j^*,\ell^*)$ under the given order.

%

We now consider iteration $i=i^*$ of the second inner loop (when $\ell=\ell^*$). We have shown that $(f_{j^*},L_{j^*},R_{j^*},{j^*}) \in D_{\ell^*}$  and we have that $f'_{i^*}=f_{j^*}$ (so $k=j^*$) but the underlying strings are not equal. If $\ell=0$ then we also have that $T[j^*\tau \ldots (j^*+1)\tau-1]$ is in the $AC$ dictionary. Therefore inspecting the current $AC$ state, will cause an abort (lines 20-21). If $\ell>0$ then as $(i^*,j^*,\ell^*)$ is minimal, either $L'_{i^*} \neq L_{j^*}$ or $R'_{i^*} \neq R_{j^*}$ which again causes an abort (line 18),  concluding the correctness.

\paragraph{Time-Space Complexity} We begin by upper bounding the space used and the time taken to performs all dictionary operations on $D_\ell$ for any $\ell$. First observe that there are at most $O(n/\tau)$ insertions (line 8) and at most $O(n)$ look-up operations (lines 7,16). We choose the dictionary data structure employed based on the relationship between $n$ and $\tau$. If $\tau > \sqrt{n}$ then we use the deterministic dynamic dictionary of Ru\v{z}i\'{c}~\cite{Ruzic2008}. Using the correct choice of constants, this dictionary supports look-ups and insert operations in $O(1)$ and $O(\sqrt{n})$ time respectively (and linear space). As there are only $O(n/\tau) = O(\sqrt{n})$ inserts, the total time taken is $O(n)$ and the space used is $O(n/\tau)$. If $\tau \leq \sqrt{n}$ we use the Las-Vegas dynamic dictionary of Dietzfelbinger and  Meyer auf der Heide~\cite{Dietzfelbinger1990}. If $\Theta(\sqrt{n}) = O(n/\tau)$ space is used for $D_\ell$, as we perform $O(n)$ operations, every operation takes $O(1)$ time with high probability. In either case, over all $\ell$ we take $O(n \log n)$ total time processing dictionary operations.

The operations performed on $AC$ fall conceptually into three categories, each totalling $O(n \log n)$ time. First we insert $O(n/\tau)$ $\tau$-length substrings into the $AC$ dictionary (line 9). Second, we feed $T$ into the automaton (line 11,15) and third, we inspect the $AC$ state at most $n$ times (line 20). The space to store $AC$ is $O(n)$, the total length of the substrings.

 Finally we bound the time spent constructing fingerprints. We first consider computing $f'_i$ (line 13) for $i>1$. We can compute $f'_i$ in $O(1)$ time from $f'_{i-1}$,  $T[i-1]$ and $T[i+\tau\cdot 2^\ell]$. This follows immediately from the definition of $\phi$. We can compute $L'_i$ and $R'_i$ analogously. Over all $i,\ell$, this gives $O(n \log n)$ time. Similarly we can compute $f_j$ from $f_{j-1}$, $T[(j-1)\tau \ldots j\tau-1]$ and $T[(j-1+2^\ell)\tau \ldots (j+2^\ell)-1]$ in $O(\tau)$ time. Again this is analogous for $L'_i$ and $R'_i$. Summing over all $j,\ell$ this gives $O(n \log n)$ time again. Finally observe that the algorithm only needs to store the current and previous values for each fingerprint so this does not dominate the space usage.

\section{Longest Common Extensions on Two Strings}
We now show how to efficiently reduce LCE queries between two strings to LCE queries on a single string. We generalise our notation as follows. Let $P$ and $T$ be strings of lengths $n$ and $m$, respectively. Define $\LCE_{P,T}(i,j)$ to be the length of the longest common prefix of the substrings of $P$ and $T$ starting at $i$ and $j$, respectively. For a single string $P$, we define $\LCE_P(i,j)$ as usual. We can always trivially solve the LCE problem for $P$ and $T$ by solving it for the string obtained by concatenating $P$ and $T$. We show the following improved result.
\begin{theorem}\label{thm:twostringLCE}
Let $P$ and $T$ be strings of lengths $m$ and $n$, respectively. Given a solution to the LCE problem on $P$ using $s(m)$ space and $q(m)$ time and a parameter $\tau$, $1 \leq \tau \leq n$, we can solve the LCE problem on $P$ and $T$ using $O(\frac{n}{\tau}+s(m))$ space and $O(\tau+q(m))$ time.
\end{theorem}
For instance, plugging in Theorem~\ref{thm:main}(i) in Theorem~\ref{thm:twostringLCE} we obtain a solution using $O(\frac{n}{\tau} + \frac{m}{\sqrt{\tau}})$ space and $O(\tau)$ time. Compared with the direct solution on the concatenated string that uses $O(\frac{n + m}{\sqrt{\tau}})$ we save substantial space when $m \ll n$.

\subsection{The Data Structure}
The basic idea for our data structure is inspired by a trick for reducing constant factors in the space for the LCE data structures~\cite[Ch.\ 9.1.2]{Gusfield1997}. We show how to extend it to obtain asymptotic improvements. Let $P$ and $T$ be strings of lengths $m$ and $n$, respectively. Our data structure for LCE queries on $P$ and $T$ consists of the following information.
\begin{itemize}
\item A data structure that supports LCE queries for $P$ using $s(m)$ space and $q(m)$ query time.
\item An array $A$ of length $\floor{\frac{n}{\tau}}$ such that $A[i]$ is the maximum $\LCE$ value between any suffix of $P$ and the suffix of $T$ starting at position $i\cdot \tau$, that is, $A[i] = \operatorname*{max}_{j=1 \ldots m} \LCE_{P,T}(j,i\tau)$.
\item An array $B$ of length $\floor{\frac{n}{\tau}}$ such that $B[i]$ is the index in $P$ of a suffix that maximises the $\LCE$ value, that is,
$B[i] =  \operatorname*{arg\,max}_{j=1 \ldots m} \LCE_{P,T}(j,i\tau)$.
\end{itemize}
Arrays $A$ and $B$ use $O(n/\tau)$ space and hence the total space is $O(n/\tau + s(m))$. We answer an $\LCE_{P,T}$ query as follows. Suppose that $\LCE_{P,T}(i,j) < \tau$. In that case we can determine the value of $\LCE_{P,T}(i,j)$ in $O(\tau)$ time by explicitly comparing the characters from position $i$ in $P$ and $j$ in $T$ until we encounter the mismatch. If $\LCE_{P,T}(i,j) \geq \tau$, we explicitly compare $k < \tau$ characters until $j+k \equiv 0 \pmod{\tau}$. When this occurs we can lookup a suffix of $P$, which the suffix $j+k$ of $T$ follows at least as long as it follows the suffix $i+k$ of $P$. This allows us to reduce the remaining part of the $\LCE_{P,T}$ query to an $\LCE_P$ query between these two suffixes of $P$ as follows.
\[
\LCE_{P,T}(i,j) = k + \min \left( A\left [\frac{j+k}{\tau} \right], \LCE_P \left (i+k, \, B \left [\frac{j+k}{\tau} \right] \right)\right)\; .
\]
We need to take the minimum of the two values, since, as shown by \autoref{ex:LCEPT}, it can happen that the $\LCE$ value between the two suffixes of $P$ is greater than that between suffix $i+k$ of $P$ and suffix $j+k$ of $T$. In total, we use $O(\tau + q(m))$ time to answer a query. This concludes the proof of Theorem~\ref{thm:twostringLCE}.
\begin{example}\label{ex:LCEPT}
Consider the query $\LCE_{P,T}(2,13)$ on the string $P$ from \autoref{ex:suffixsample} and
\begin{center}
\includegraphics{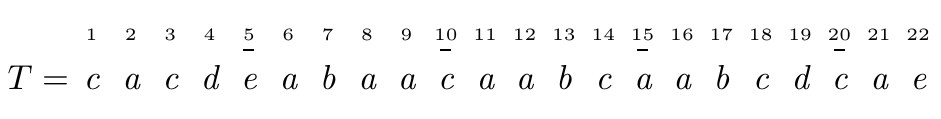}

\end{center}
The underlined positions in $T$ indicate the positions divisible by $5$. As shown below, we can use the array $A=[0,6,4,2]$ and $B=[16,3,11,10]$.
\begin{center}
\includegraphics{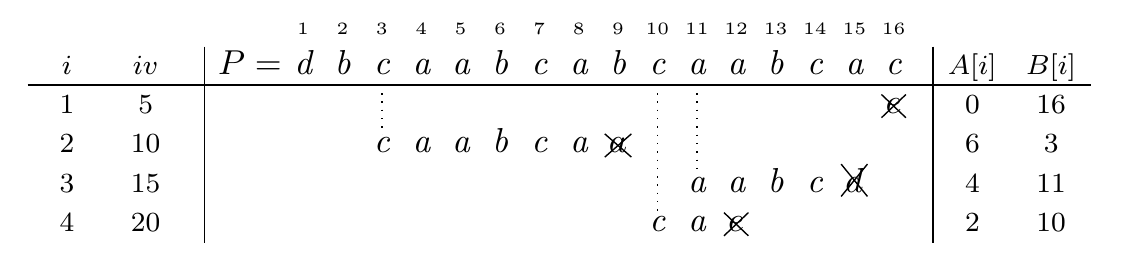}
\end{center}
To answer the query $\LCE_{P,T}(2,13)$ we make $k=2$ character comparisons and find that
\begin{align*}
\LCE_{P,T}(2,13) &= 2+\min \left( A\left [ \frac{13+2}{5} \right], \LCE_P \left (2+2, B\left [ \frac{13+2}{5} \right ] \right ) \right) \\&= 2 + \min(4,5) = 6 \; .
\end{align*}
\end{example}

\section{Lower Bound}

In this section we prove the lower bound for \LCE data structures.

\begin{lemma}\label{lem:lowerbound}
In the non-uniform cell probe model any \LCE data structure that uses $O(n/\tau)$ bits additional space for an input array of size $n$, requires $\Omega(\tau)$ query time, for any $\tau$, where $1 \leq \tau \leq n$.
\end{lemma}

\paragraph{Range Minimum Queries (RMQ)} Given an array $A$ of integers, a range minimum query data structure must support queries of the form:
\begin{itemize}
\item  RMQ($l$, $r$): Return the position of the minimum element in $A[l,r]$. 
\end{itemize} 
Brodal et al.~\cite{BDR12} proved any algorithm that uses $n/\tau$ bits additional space to solve the RMQ problem for an input array of size $n$ (in any dimension), requires $\Omega(\tau)$ query time, for any $\tau$, where $1 \leq \tau \leq n$, 
even in a binary array $A$ consisting only of 0s and 1s. Their proof is in the non-uniform cell probe model~\cite{PB99}. In this model, computation is free, and time is counted as the number of cells accessed (probed) by the query algorithm. The algorithm is allowed to be non-uniform, i.e. for different values of input parameter $n$, we can have different algorithms.

To prove \autoref{lem:lowerbound}, we will show that any \LCE data structure can be used to support range minimum queries, using one \LCE query and $O(1)$ space additional to the space of the \LCE data structure. 
 
\paragraph{Reduction} Using any data structure supporting \LCE queries we can support RMQ queries on a binary array A as follows:
In addition to the \LCE data structure, store the indices $i$ and $j$ of the longest substring $A[i,j]$ of $A$ consisting of only 1's.
To answer a query RMQ(l,r) compute $\res = \LCE(l,i)$. Let $z = j-i+1$ denote the length of the longest substring of $1$'s. Compare $\res$ and $z$:
\begin{itemize}
\item If $\res \leq z$ and $l+ \res \leq r$, return $l + \res$. 
\item If $\res > z$ and $l+z \leq r$,  return $l + z$. 
\item Otherwise return any position in $[l,r]$.
\end{itemize}
To see that this correctly answers the RMQ query consider the two cases. 
If  $\res \leq z$ then $A[l, l + \res-1]$ contains only 1's, since $A[i, j]$ contains only 1's. Thus  position $l + \res$ is the index of the first 0 in $A[l,l + \res]$.  It follows that if $l+\res \leq r$, then $A[l+\res]=0$ and is a minimum in $A[l,r]$. Otherwise, $A[l,r]$ contains only 1's.

If $\res > z$ then $A[l, l+z-1]$ contains only 1's and position $A[l+z]=0$ . There are two cases: Either  $l+z \leq r$, in which case this position contains the first 0 in $A[l,r]$. Or  $l+z > r$ in which case $A[l,r]$ contains only 1's.

\section{Conclusions and Open Problems}

We have presented new deterministic and randomised time-space trade-offs for the Longest Common Extension problem. In particular, we have shown that there is a data structure for LCE queries using $O(n/\tau)$ space and supporting LCE queries in $O(\tau \log (|\LCE(i,j)|/\tau))$ time. We have also shown that any LCE data structure using $O(n/\tau)$ bits of space must have query time $\Omega(\tau)$. Consequently, the time-space product of our trade-off is essentially a factor $\log^2 n$ from optimal. It is an interesting open problem whether this gap can be closed.

Another open question, which is also of general interest in applications of error-free fingerprinting, is whether it is possible to find a $\tau$-good fingerprinting function on a string of length $n$ in $O(n\log n)$ time with high probability and $O(n/ \tau)$ space simultaneously. Moreover, a deterministic way of doing this would provide a strong tool for derandomising solutions using fingerprints, including the results in this paper.

\paragraph{Acknowledgements}
We thank the anonymous referes for their valuable comments, and in particular for pointing out an issue with the preprocessing space of the randomised solution.


\bibliographystyle{abbrv}
\bibliography{paper}

\end{document}